\documentclass[letterpaper,12pt]{article}
\usepackage{amssymb, amsmath}
\usepackage{amsthm}
\usepackage{algorithm}
\usepackage{algorithmic}
\usepackage{enumerate}
\usepackage[numbers,square,comma,sort&compress]{natbib}
\usepackage[pdftex,colorlinks]{hyperref}
\usepackage{a4}

\newtheorem{theorem}{Theorem}
\newtheorem{corollary}[theorem]{Corollary}

\newtheorem{lemma}[theorem]{Lemma}

\newtheorem*{intuition}{An intuitive exposition}
\newtheorem*{restate1}{Restatement of Lemma~\ref{recurrencelemma1}}
\newtheorem*{restate2}{Restatement of Lemma~\ref{recurrencelemma2}}
\newcommand{\comment}[1]{}
 
\newcommand\bs[1]{\boldsymbol{#1}}

\begin{document}
\title{Metric $1$-median selection with fewer queries \footnote{A preliminary version of
this paper appears in {\em Proceedings of the 2017 IEEE International Conference on Applied System Innovation},
Sapporo, Japan.}}

\author{
Ching-Lueh Chang \footnote{Department of Computer Science and
Engineering,
Yuan Ze University, Taoyuan, Taiwan. Email:
clchang@saturn.yzu.edu.tw}
}


\maketitle

\begin{abstract}
Let
$h\colon\mathbb{Z}^+\to\mathbb{Z}^+\setminus\{1\}$
be
any function
such that
$h(n)$ and
$\lceil n^{1/h(n)}\rceil$
are computable from $n$ in $O(h(n)\cdot n^{1+1/h(n)})$ time.
We show
that given
any $n$-point
metric space
$(M,d)$,
the problem of finding
$\mathop{\mathrm{argmin}}_{i\in M}\,\sum_{j\in M}\,d(i,j)$
(breaking ties arbitrarily)
has a deterministic, $O(h(n)\cdot n^{1+1/h(n)})$-time,
$O(n^{1+1/h(n)})$-query, $(2\,h(n))$-approximation and nonadaptive algorithm.
Our proofs modify those of Chang~\cite{Cha15, Cha15CMCT} with the
following improvements:
\begin{itemize}
\item We improve Chang's~\cite{Cha15} query complexity
of $O(h(n)\cdot n^{1+1/h(n)})$ to $O(n^{1+1/h(n)})$, everything else
being equal.
\item Chang's~\cite{Cha15CMCT} unpublished work establishes
our result only when $n$ is a perfect $(h(n))$th power.
\end{itemize}
\end{abstract}

\section{Introduction}

A
metric space
is
a nonempty
set $M$ endowed with a function
$d\colon M\times M\to[\,0,\infty\,)$ such that
\begin{eqnarray*}
d\left(x,x\right)&=&0,\\
d\left(x,y\right)&>&0,\\
d\left(x,y\right)&=&d\left(y,x\right),\\
d\left(x,y\right)+d\left(y,z\right)&\ge& d\left(x,z\right)
\end{eqnarray*}
for all distinct $x$, $y$,
$z\in M$.
Given an $n$-point metric space
$(M,d)$,
{\sc metric $1$-median} asks for
$\mathop{\mathrm{argmin}}_{i\in M}\,\sum_{j\in M}\,d(i,j)$,
breaking ties arbitrarily.
As usual, an
algorithm for {\sc metric $1$-median}
may make any
query $(x,y)\in M^2$ to obtain $d(x,y)$ in $O(1)$ time.
By convention,
if
its sequence of queries
depends only on $M$ but not on $d$,
then it is said to be nonadaptive.
For all $c\ge1$,
a $c$-approximate $1$-median of $(M,d)$
refers to a point $z\in M$
satisfying
$$\sum_{x\in M}\,d\left(z,x\right)\le c\cdot \min_{y\in M}\,\sum_{x\in M}\,d\left(y,x\right).$$
A $c$-approximation algorithm for {\sc metric $1$-median}
is one that always outputs a $c$-approximate $1$-median.
As usual for metric-space problems,
we adopt the real RAM model.
In particular, we need $O(1)$-time real-number additions, multiplications and comparisons.

{\sc Metric $1$-median}
has
a Monte-Carlo
$O(n/\epsilon^2)$-time $(1+\epsilon)$-approximation algorithm
for all
$\epsilon>0$~\cite{Ind99, Ind00}.
Kumar et al.~\cite{KSS10} give
a Monte-Carlo
$O(D\cdot\exp{(1/\epsilon^{O(1)})})$-time
$(1+\epsilon)$-approximation algorithm
for
{\sc metric $1$-median}
in $\mathbb{R}^D$, where
$\epsilon>0$ and $D\in\mathbb{Z}^+$.
Algorithms abound for the more general metric $k$-median problem~\cite{KSS10, GMMMO03, Che09}.
In social network analysis, the closeness centrality of an actor
is the reciprocal of the
average distance from
that actor
to others~\cite{WF94}.
It is a popular measure of importance~\cite{WF94}.
So {\sc metric $1$-median}
can be interpreted as finding
one of the most important actors.
Furthermore, {\sc metric $1$-median}
generalizes
the classical median selection~\cite{CLRS09}.

This paper focuses
on deterministic sublinear-time algorithms
for {\sc metric $1$-median}, where ``sublinear'' means ``$o(n^2)$''
by convention
because there are
$n(n-1)/2$ nonzero distances~\cite{Ind99}.
In particular, we shall improve the following theorem:
\begin{theorem}[{Implicit in~\cite{Cha15, Wu14}}]
\label{oldresultgeneralized}
Let
$h\colon\mathbb{Z}^+\to\mathbb{Z}^+\setminus\{1\}$
be
any function
such that
$h(n)$ and
$\lceil n^{1/h(n)}\rceil$
are computable from $n$ in $O(h(n)\cdot n^{1+1/h(n)})$ time.
Then {\sc metric $1$-median} has a
deterministic, $O(h(n)\cdot n^{1+1/h(n)})$-time,
$O(h(n)\cdot n^{1+1/h(n)})$-query,
$(2\,h(n))$-approximation and nonadaptive algorithm.
\end{theorem}
Theorem~\ref{oldresultgeneralized} is rather general:
Together with
Chang's
lower bound,
it
implies
{\footnotesize 
\begin{eqnarray}
&&
\min\left\{c\ge1\mid \text{{\sc metric $1$-median} has a deterministic
$O(n^{1+\epsilon})$-query $c$-approx.\ algorithm}\right\}\nonumber\\
&=&\min\left\{c\ge1\mid \text{{\sc metric $1$-median} has a deterministic
$O(n^{1+\epsilon})$-time $c$-approx. algorithm}\right\}\nonumber\\
&=&2\left\lceil\frac{1}{\epsilon}\right\rceil\nonumber
\end{eqnarray}
}
for {\em all} constants
$\epsilon\in(0,1)$~\cite{Cha18}.

To prove Theorem~\ref{oldresultgeneralized},
Chang~\cite{Cha15}
designs a
function
$\tilde{d}\colon \{0,1,\ldots,n-1\}^2\to[\,0,\infty\,)$
such that
a $1$-median
w.r.t.\
$\tilde{d}$
is $(2\,h(n))$-approximate w.r.t.\ $d$ and is computable
in $O(h(n)\cdot n^{1+1/h(n)})$ time.
However, $\tilde{d}(\cdot,\cdot)$ depends on
$\Theta(h(n)\cdot n^{1+1/h(n)})$ distances of $d$,
forbidding us to improve the query complexity of
$O(h(n)\cdot n^{1+1/h(n)})$
in
Theorem~\ref{oldresultgeneralized}.
Wu's~\cite{Wu14} algorithm
also
makes $\Theta(h(n)\cdot n^{1+1/h(n)})$ queries.
In contrast,
we
design
a new
function,
$\hat{d}$,
that depends on only
$\Theta(n^{1+1/h(n)})$ distances of $d$
and is otherwise similar to Chang's $\tilde{d}$.
This results in a
deterministic, $O(h(n)\cdot n^{1+1/h(n)})$-time,
$O(n^{1+1/h(n)})$-query,
$(2\,h(n))$-approximation and nonadaptive algorithm
for {\sc metric $1$-median},
shaving the factor of $h(n)$ from
the query complexity in
Theorem~\ref{oldresultgeneralized}.
So our result
strengthens
Theorem~\ref{oldresultgeneralized} {\em whenever}
$h(n)=\omega(1)$.\footnote{For a
concrete example, if we look for
$\tilde{O}(n)$-time algorithms,
then we will need $h(n)=\omega(1)$ in Theorem~\ref{oldresultgeneralized}.}
The idea of our construction of $\hat{d}$
comes from an unpublished
paper of Chang~\cite{Cha15CMCT}.
Aside from
the
design of $\hat{d}$,
our
proofs modify
those of
Chang~\cite{Cha15}
by cumbersome brute force.

As a corollary to
our result,
{\sc metric $1$-median} has
a deterministic, $O(\exp(O(1/\epsilon))\cdot n\log n)$-time,
$O(\exp(O(1/\epsilon))\cdot n)$-query, $(\epsilon\log n)$-approximation and nonadaptive algorithm
for all constants $\epsilon>0$.
For each constant $D\ge1$,
{\sc metric $1$-median}
in $\mathbb{R}^D$
has
a deterministic,
$O(k(\epsilon)\cdot n\log n)$-time,
$O(k(\epsilon)\cdot n)$-space
and $(1+\epsilon)$-approximation algorithm
for
some
function
$k\colon (\,0,\infty\,)\to(\,0,\infty\,)$~\cite[Theorem~9]{BMM03}.
But for general metrics,
we do not even know whether
deterministic
$O(n\log n)$-time
algorithms can be $o(\log n)$-approximate.

Our
improvement of the query complexity in
Theorem~\ref{oldresultgeneralized} to $O(n^{1+1/h(n)})$
almost matches
the following
lower bound,
leaving
a multiplicative gap of
about
$h(n)$
on the query complexity:
\begin{theorem}[\cite{Cha18}]\label{currentlybestlowerbound}
{\sc Metric $1$-median} has no deterministic $o(n^{1+1/(h(n)-1)}/h(n))$-query
$(2\,h(n)\cdot(1-\epsilon))$-approximation algorithms for any constant $\epsilon>0$
and any $h\colon \mathbb{Z}^+\to
\mathbb{Z}^+\setminus\{1\}$ satisfying $h=o(n^{1/(h(n)-1)})$.
\end{theorem}
\comment{ 
that {\sc metric $1$-median}
has no deterministic
$o(n^{1+1/(h(n)-1)}/h(n))$-query $(2\,h(n)-\epsilon)$-approximation
algorithms for any constant $\epsilon>0$ and any $h\colon\mathbb{Z}^+\to\mathbb{Z}^+\setminus\{1\}$
satisfying $h(n)=o(n^{1/(h(n)-1)})$.
}


Although
all deterministic algorithms for {\sc metric $1$-median}
are
provably
outperformed
by
Indyk's randomized algorithm~\cite{Cha18},
their limits are
nonetheless
worth studying because
randomness is a computational resource from the viewpoint of
theoretical computer science.
For example, after oblivious permutation routing on the $n$-node hypercube
is
known
to have
a randomized $O(\log n)$-time algorithm~\cite{Val82},
its deterministic time complexity of $\Omega(\sqrt{n}/{\log n})$
is still considered interesting~\cite{KKT91}.
Similarly, the deterministic
time complexities
of
primality testing
and polynomial identity testing are
considered important (and the latter still open) even after
their randomized counterparts are known to be efficient~\cite{AKS04, KI04}.

Before diving into details, we give an intuitive exposition
as to why we have an $O(n^{1+1/h(n)})$-query $(2\,h(n))$-approximation algorithm.

\begin{intuition}
Take $h\equiv 3$ for an example.
We now
intuitively explain
the approximation ratio
of $2h=6$.
For simplicity, assume
$n$ to be cubic and let
$t=n^{1/3}$.
Clearly,
\begin{eqnarray}
\left\{0,1,\ldots,n-1\right\}=\left\{s_2 t^2+s_1 t+s_0\mid
s_2,s_1,s_0\in\left\{0,1,\ldots,t-1\right\}\right\}.
\label{exactexpressions}
\end{eqnarray}
For all $i\in\{0,1,\ldots,n-1\}$ and $s_2$, $s_1$, $s_0\in\{0,1,\ldots,t-1\}$,
\begin{eqnarray}
&&\hat{d}
\left(i,it^3+s_2 t^2+s_1 t + s_0\bmod{n}\right)\nonumber\\
&\equiv&
d\left(i,it+s_2\bmod{n}\right)\nonumber\\
&+& d\left(it+s_2\bmod{n},it^2+s_2 t + s_1\bmod{n}\right)\nonumber\\
&+& d\left(it^2+s_2 t + s_1\bmod{n},it^3+s_2 t^2+s_1 t+s_0\bmod{n}\right).
\label{fakepseudodistance}
\end{eqnarray}
By
Eq.~(\ref{exactexpressions}), the domain of $\hat{d}$
is $\{0,1,\ldots,n-1\}^2$.

Let $\bs{u}$ and $\bs{v}$  be independent and uniformly random elements
of $\{0,1,\ldots,n-1\}$.
Assuming
$\mathop{\mathrm{gcd}}(t,n)=1$ for now,
it is nontrivial but provable that
$$
\hat{d}\left(\bs{u},\bs{v}\right)
=d\left(\bs{u},\bs{c}_1\right)
+d\left(\bs{c}_1,\bs{c}_2\right)
+d\left(\bs{c}_2,\bs{v}\right)
$$
for {\em uniformly random}
(but dependent) elements $\bs{c}_1$ and $\bs{c}_2$
of $\{0,1,\ldots,n-1\}$.
So for a $1$-median {\rm OPT} w.r.t.\ $d$,
the triangle inequality implies
\begin{eqnarray*}
\hat{d}\left(\bs{u},\bs{v}\right)
&\le&
d\left(\text{\rm OPT},\bs{u}\right)+d\left(\text{\rm OPT},\bs{c}_1\right)\\
&+&
d\left(\text{\rm OPT},\bs{c}_1\right)+d\left(\text{\rm OPT},\bs{c}_2\right)\\
&+&
d\left(\text{\rm OPT},\bs{c}_2\right)+d\left(\text{\rm OPT},\bs{v}\right),
\end{eqnarray*}
whose right-hand side sums $6$ distances from {\rm OPT} to uniformly random points.
Now take
expectations on both sides
to
see that
the average $\hat{d}$-distance
is
at most
$6$ times the average $d$-distance from OPT to all points.
This hints that a $1$-median w.r.t.\ $\hat{d}$ is a $6$-approximate $1$-median
w.r.t.\ $d$.

To
find a $1$-median w.r.t.\ $\hat{d}$, we shall
find
$\sum_{j=0}^{n-1}\,\hat{d}(i,j)$ for all $i\in\{0,1,\ldots,n-1\}$.
By Eq.~(\ref{exactexpressions}),
\begin{eqnarray*}
\sum_{j=0}^{n-1}\,\hat{d}\left(i,j\right)=
\sum_{s_2,s_1,s_0=0}^{t-1}\,
\hat{d}
\left(i,it^3+s_2 t^2+s_1 t + s_0\bmod{n}\right).
\end{eqnarray*}
Furthermore, by Eq.~(\ref{fakepseudodistance}),
\begin{eqnarray}
&&\sum_{s_2,s_1,s_0=0}^{t-1}\,
\hat{d}
\left(i,it^3+s_2 t^2+s_1 t + s_0\bmod{n}\right)\nonumber\\
&=&
\sum_{s_2,s_1,s_0=0}^{t-1}\,
\left[
d\left(i,it+s_2\bmod{n}\right)
\phantom{d\left(it^2+s_2 t + s_1\bmod{n},it^3+s_2 t^2+s_1 t+s_0\right)}
\right.\nonumber\\
&&+ d\left(it+s_2\bmod{n},it^2+s_2 t + s_1\bmod{n}\right)\nonumber\\
&&+
\left.
d\left(it^2+s_2 t + s_1\bmod{n},it^3+s_2 t^2+s_1 t+s_0\bmod{n}\right)
\right].
\label{illustrationofthreelevelsum}
\end{eqnarray}
Albeit nontrivial,
the right-hand side of Eq.~(\ref{illustrationofthreelevelsum})
is a
$3$-level sum (so called because it is taken over three variables $s_2$, $s_1$ and $s_0$)
expressible
using $2$-level sums
of a similar form.
So, unsurprisingly,
it
can
be found by dynamic programming (that builds up sums with increasing
levels).


Now comes the key question:
How
do we shave the factor of $h$
from
the query complexity in
Theorem~\ref{oldresultgeneralized}?
Observe that $\hat{d}$ in
Eq.~(\ref{fakepseudodistance}) depends only on distances of the form
$d(j,jt+s\bmod{n})$, where $j\in\{0,1,\ldots,n-1\}$ and $s\in\{0,1,\ldots,t-1\}$,
for a total of $nt$ distances.
Instead, Chang's~\cite{Cha15} $\tilde{d}$
is
\begin{eqnarray*}
&&
\tilde{d}\left(i,\, i+s_2t^2+s_1t+s_0\bmod{n}\right)\\
&\equiv& d\left(i,\, i+s_2t^2\bmod{n}\right)\\
&+& d\left(i+s_2t^2\bmod{n},\, i+s_2t^2+s_1t\bmod{n}\right)\\
&+& d\left(i+s_2t^2+s_1t\bmod{n},\, i+s_2t^2+s_1t+s_0\bmod{n}\right)
\end{eqnarray*}
for all
$i\in\{0,1,\ldots,n-1\}$ and $s_2$, $s_1$, $s_0\in\{0,1,\ldots,t-1\}$.
Now verify
$\tilde{d}$
to
depend
on all distances of the form
$d(j,j+st^k\bmod{n})$, where $j\in\{0,1,\ldots,n-1\}$, $s\in\{0,1,\ldots,t-1\}$
and $k\in\{0,1,2\}$---there are
$3nt=hnt$ such distances.
In summary,
our
$\hat{d}$
depends on $h$ times
fewer distances
than Chang's $\tilde{d}$ does!
\comment{ 
Aside from the design of $\hat{d}$, our proofs
modify those of
Chang~\cite{Cha15} in a very cumbersome but non-insightful way.
}

Everything so far is equivalent to Chang's~\cite{Cha15CMCT} unpublished work, which assumes
$n$ to be cubic (or, more generally, a perfect $h$th power)
for Eq.~(\ref{exactexpressions}) to hold---Gladly,
this assumption can be removed by slightly modifying
Chang's~\cite{Cha15} dynamic-programming approach.

We have assumed $\mathop{\mathrm{gcd}}(t,n)=1$, which may be false.
To get around, our construction takes suitable $t=\Theta(n^{1/h})$
and $\sigma\in\{0,1\}$
satisfying $\mathop{\mathrm{gcd}}(t,n-\sigma)=1$.
Allowing $\sigma$ to be $1$
makes
our
actual
$\hat{d}$
slightly
different
from that in Eq.~(\ref{fakepseudodistance}) and complicates our proofs
significantly.
\comment{ 
Assume $n$ to be a perfect $(h(n))$th power
and treat $0$, $1$, $\ldots$, $n-1$
as elements in $S^h$ for a set $S$.
Then Chang's~\cite{Cha15} pseudo-distance function is
\begin{eqnarray*}
&&\tilde{d}\left(\left(u_1,u_2,\ldots,u_h\right),
\left(v_1,v_2,\ldots,v_h\right)\right)\\
&=&
\sum_{i=0}^{h-1}\,
d\left(\left(v_1,v_2,\ldots,v_i,u_{i+1},u_{i+2},\ldots,u_h\right),
\left(v_1,v_2,\ldots,v_{i+1},u_{i+2},u_{i+3},\ldots,u_h\right)
\right)
\end{eqnarray*}
for all $u_1$, $u_2$, $\ldots$, $u_h$, $v_1$, $v_2$, $\ldots$,
$v_h\in S$.
}
\end{intuition}

\comment{ 
We reduce the query complexity
in this result to $O(n^{1+1/h(n)})$.

Chang's~\cite{Cha15} approximation algorithm for {\sc metric $1$-median}
uses a
function
$\tilde{d}\colon \{0,1,\ldots,n-1\}^2\to[0,\infty)$
such that
a $1$-median
w.r.t.\
$\tilde{d}$
is $(2\,h(n))$-approximate w.r.t.\ $d$ and is computable
in $O(h(n)\cdot n^{1+1/h(n)})$ time.
However, $\tilde{d}(\cdot,\cdot)$ depends on
$\Omega(h(n)\cdot n^{1+1/h(n)})$ distances of $d$.
Instead, we design a new
function,
$\hat{d}$,
that depends on only
$\Theta(n^{1+1/h(n)})$ distances of $d$.
}

\section{The new pseudo-distance function $\hat{d}$}
\label{pseudometric}

Let
$(\{0,1,\ldots,n-1\},d)$ be a metric space and
$h\colon \mathbb{Z}^+\to\mathbb{Z}^+\setminus\{1\}$
be a computable function.
By
Bertrand's postulate,
there exists a prime number
$t
\in[\,\lceil n^{1/h(n)}\rceil,\, 2\cdot \lceil n^{1/h(n)}\rceil\,]$.
Clearly, $\mathop{\mathrm{gcd}}(n-1,n)=1$.
So the primality of $t$ implies the existence of
$\sigma\in\{0,1\}$ such that
$\mathop{\mathrm{gcd}}(t,n-\sigma)=1$.
For convenience, $h\equiv h(n)$.
For all $j\in \{0,1,\ldots,n-1\}$,
write
\begin{eqnarray}
\left(s_{h-1}(j),s_{h-2}(j),\ldots,s_0(j)\right)
\in\left\{0,1,\ldots,t-1\right\}^h
\nonumber
\end{eqnarray}
for
the
unique
$t$-ary representation of $j$,
following Chang~\cite{Cha15}.\footnote{As $t\ge \lceil n^{1/h}\rceil$,
the $t$-ary representation of any of $0,1,\ldots,n-1$ has at most $h$ digits.}
So
\begin{eqnarray}
\sum_{\ell=0}^{h-1}\,s_{h-1-\ell}(j)\cdot t^{h-1-\ell}=j.
\label{multiaryrepresentation}
\end{eqnarray}
For any predicate $P$,
let $\chi[\,P\,]=1$ if $P$ is true and $\chi[\,P\,]=0$ otherwise.

Define
\begin{eqnarray}
d^{(n-\sigma)}\left(x,y\right)
\equiv
d\left(x\bmod{\left(n-\sigma\right)},\,
y\bmod{\left(n-\sigma\right)}\right)
\label{themodularversionofdistances}
\end{eqnarray}
for all $x$, $y\in\mathbb{N}$.
Clearly,
$d^{(n-\sigma)}$
is symmetric and obeys
the triangle inequality,
just like $d$.
For all $i$, $j\in\{0,1,\ldots,n-\sigma-1\}$,
define
{\small 
\begin{eqnarray}
&&
\hat{d}\left(i,it^h+j\bmod{\left(n-\sigma\right)}\right)
\nonumber\\
&\equiv&
\sum_{k=0}^{h-1}\, d^{(n-\sigma)}
\left(it^k+\sum_{\ell=0}^{k-1}\,s_{h-1-\ell}(j)\cdot t^{k-1-\ell},
\,
it^{k+1}+\sum_{\ell=0}^k\,s_{h-1-\ell}(j)\cdot t^{k-\ell}\right).
\,\,\,\,\,\,\,\,\,\,
\label{pseudodistance}
\end{eqnarray}
}
This
and the triangle inequality for $d^{(n-\sigma)}$
imply
$$
\hat{d}\left(i,it^h+j
\bmod{\left(n-\sigma\right)}\right)
\ge
d^{(n-\sigma)}
\left(i,\, it^h+\sum_{\ell=0}^{h-1}\, s_{h-1-\ell}(j)\cdot
t^{h-1-\ell}\right).
$$
\footnote{Note that $it^k+\sum_{\ell=0}^{k-1}\,s_{h-1-\ell}(j)\cdot t^{k-1-\ell}=i$ when
$k=0$ and $it^{k+1}+\sum_{\ell=0}^k\,s_{h-1-\ell}(j)\cdot t^{k-\ell}=
it^h+\sum_{\ell=0}^{h-1}\, s_{h-1-\ell}(j)\cdot
t^{h-1-\ell}$ when $k=h-1$.}
So by
Eq.~(\ref{multiaryrepresentation}),
\begin{eqnarray}
\hat{d}\left(i,it^h+j\bmod{\left(n-\sigma\right)}\right)
\ge
d^{(n-\sigma)}
\left(i,it^h+j\right)
\label{pseudodistanceislarger}
\end{eqnarray}
for all $i$, $j\in \{0,1,\ldots,n-\sigma-1\}$.

Having defined
$\hat{d}(i,it^h+j\bmod{(n-\sigma)})$
in
Eq.~(\ref{pseudodistance})
for all $i$, $j\in \{0,1,\ldots,n-\sigma-1\}$,
the domain of
$\hat{d}$
is
$\{0,1,\ldots,n-\sigma-1\}^2$.\footnote{Note that each pair in
$\{0,1,\ldots,n-\sigma-1\}^2$
can be written as $(i,it^h+j\bmod{(n-\sigma)})$
for a unique pair $(i,j)\in \{0,1,\ldots,n-\sigma-1\}^2$.}
\comment{ 
So for convenience,
interpret $\hat{d}(x,y)$ as
$$
\hat{d}\left(x\bmod{\left(n-\sigma\right)},\,
y\bmod{\left(n-\sigma\right)}\right).
$$
for all $x$, $y\in\mathbb{N}$.
}
\comment{ 
The following lemma says
that $\hat{d}$
is an upper bound on
$d$.
}
\comment{ 
Eqs.~(\ref{multiaryrepresentation})--(\ref{pseudodistance})
and
the triangle inequality for $d$
imply the following lemma.

\begin{lemma}\label{pseudodistanceislarger}
For all $i$, $j\in\{0,1,\ldots,n-1\}$,
$$
\hat{d}\left(i,it^h+j\right)
\ge
d\left(i,it^h+j\right).
$$
\end{lemma}
}
\comment{ 
\begin{proof}
By Eq.~(\ref{pseudodistance}) and the triangle inequality for $d$,
$$
\hat{d}\left(i,i+j\bmod{n}\right)
\ge
d\left(i,i+\sum_{\ell=0}^{h-1}\, s_\ell(j)\cdot
t^\ell\bmod{n}\right).
$$
This and
Eq.~(\ref{multiaryrepresentation})
complete the proof.
\end{proof}
}
Let
\begin{eqnarray}
i'=\mathop{\mathrm{argmin}}_{i=0}^{n-\sigma-1}\, \sum_{j=0}^{n-1}\,
d\left(i,j\right),
\label{theoptimalpointexcludingthelastone}
\end{eqnarray}
breaking ties arbitrarily.

\begin{intuition}
For simplicity, assume $\sigma=0$.
Pick
independent and uniformly random
elements $\bs{u}$ and $\bs{v}$ of $\{0,1,\ldots,n-1\}$.
Taking $i\leftarrow\bs{u}$ and $j\leftarrow\bs{v}$ in Eq.~(\ref{pseudodistance}),
it can be verified that each of the $h$ summands in
the right-hand side of Eq.~(\ref{pseudodistance})
is a $d$-distance
between uniformly random
points.
That is,
$$
\hat{d}\left(\bs{u},\,\bs{u}t^h+\bs{v}\bmod{n}\right)
=d\left(\bs{u},\bs{c}_1\right)+d\left(\bs{c}_1,\bs{c}_2\right)
+d\left(\bs{c}_2,\bs{c}_3\right)+\cdots+d\left(\bs{c}_{h-1},\bs{v}\right)
$$
for uniformly random (but possibly dependent) elements $\bs{c}_1$, $\bs{c}_2$, $\ldots$,
$\bs{c}_{h-1}$ of $\{0,1,\ldots,n-1\}$.
This and the triangle inequality imply
{\footnotesize 
\begin{eqnarray*}
&&\hat{d}\left(\bs{u},\,\bs{u}t^h+\bs{v}\bmod{n}\right)\\
&\le&
\left(
d\left(
i',
\bs{u}\right)+d\left(i',
\bs{c}_1\right)
\right)
+
\left(
d\left(i',
\bs{c}_1\right)+d\left(i',
\bs{c}_2\right)
\right)
+
\left(
d\left(i',
\bs{c}_2\right)+d\left(i',
\bs{c}_3\right)
\right)
+\cdots+
\left(
d\left(i',
\bs{c}_{h-1}\right)+d\left(i',
\bs{v}\right)
\right),
\end{eqnarray*}
}
whose right-hand side sums $2h$ distances from
$i'$
to uniformly random points.
Now take
expectations on both sides to see that
the average $\hat{d}$-distance is at most $2h$ times
the average $d$-distance from
$i'$
to
all
points.
But, as $\sigma=0$, $i'$ is a $1$-median w.r.t.\ $d$
by Eq.~(\ref{theoptimalpointexcludingthelastone}).
So
it is intuitive to guess that a
$1$-median w.r.t.\ $\hat{d}$
is a $(2h)$-approximate $1$-median w.r.t.\ $d$.
The
next
lemma, whose proof embeds the above idea with
technicalities,
confirms this
guess.\footnote{When
$\sigma=0$, $\alpha$ in Eq.~(\ref{theoptimalpointwithrespecttopseudodistance})
is a $1$-median w.r.t.\ $\hat{d}$, and the right-hand side of
Eq.~(\ref{label20170701_16pm_55min_b}) is $2h$ times the total $d$-distance
from
a $1$-median w.r.t.\ $d$
to all
points.
So Lemma~\ref{approximationratiolemma} implies that when $\sigma=0$, a $1$-median w.r.t.\
$\hat{d}$ is a $(2h)$-approximate $1$-median w.r.t.\ $d$.}
\end{intuition}
\comment{ 
When $\sigma=0$, the following lemma
says
that a $1$-median
w.r.t.\
$\hat{d}$
is a $(2h)$-approximate
$1$-median w.r.t.\
$d$.
}

\begin{lemma}[{cf.~\cite[Lemma~4]{Cha15}}]
\label{approximationratiolemma}
Let
\begin{eqnarray}
\alpha
=\mathop{\mathrm{argmin}}_{i=0}^{n-\sigma-1}\,
\left(
\chi\left[\,\sigma=1\,\right]\cdot d\left(i,n-1\right)
+
\sum_{j=0}^{n-\sigma-1}\,\hat{d}\left(i,\, it^h+j
\bmod{\left(n-\sigma\right)}
\right)
\right),
\label{theoptimalpointwithrespecttopseudodistance}
\end{eqnarray}
breaking ties arbitrarily.
Then
{\small 
\begin{eqnarray}
&&\sum_{j=0}^{n-1}\, d\left(\alpha,j\right)
\nonumber\\
&\le&
\chi\left[\,\sigma=1\,\right]\cdot d\left(\alpha,n-1\right)
+
\sum_{j=0}^{n-\sigma-1}\,
\hat{d}\left(\alpha,\, \alpha t^h+j
\bmod{\left(n-\sigma\right)}
\right)\label{label20170701_16pm_55min_a}\\
&\le& 2h
\cdot
\left(
\min_{i=0}^{n-\sigma-1}\,\sum_{j=0}^{n-1}\,d\left(i,j\right)
\right)
-
\chi\left[\,\sigma=1\,\right]\cdot\left(
\left(2h-1\right)\cdot d\left(i',n-1\right)
-\frac{1}{n-1}\cdot \sum_{j=0}^{n-2}\,
d\left(i',j\right)
\right).
\,\,\,\,\,\,\,\,\,\,\,\,\label{label20170701_16pm_55min_b}
\end{eqnarray}
}
\end{lemma}
\begin{proof}
Clearly,
\begin{eqnarray}
\sum_{j=0}^{n-1}\, d\left(\alpha,j\right)
&\stackrel{\text{(\ref{themodularversionofdistances})}}{=}&
\chi\left[\,\sigma=1\,\right]\cdot d\left(\alpha,n-1\right)
+
\sum_{j=0}^{n-\sigma-1}\,
d^{(n-\sigma)}\left(\alpha,j\right)\label{needequationnumberforasmallstep1} \\
&=&
\chi\left[\,\sigma=1\,\right]\cdot d\left(\alpha,n-1\right)
+
\sum_{j=0}^{n-\sigma-1}\, d^{(n-\sigma)}
\left(\alpha,\alpha
t^h+j
\right),
\,\,\,\,\,\,
\nonumber
\end{eqnarray}
where the second equality uses
Eq.~(\ref{themodularversionofdistances}) and
the one-to-one
correspondence of
$j\mapsto \alpha t^h+j\bmod{(n-\sigma)}$
for $j\in\{0,1,\ldots,n-\sigma-1\}$.

Pick
$\bs{u}$
from
$\{0,1,\ldots,n-\sigma-1\}$
uniformly at random.
Then
{\footnotesize 
\begin{eqnarray}
&&
\chi\left[\,\sigma=1\,\right]\cdot d\left(\alpha,n-1\right)
+
\sum_{j=0}^{n-\sigma-1}\,
d^{(n-\sigma)}\left(\alpha, \alpha t^h+j
\right)
\nonumber\\
&\stackrel{\text{(\ref{pseudodistanceislarger})}}{\le}&
\chi\left[\,\sigma=1\,\right]\cdot d\left(\alpha,n-1\right)
+
\sum_{j=0}^{n-\sigma-1}\,
\hat{d}\left(\alpha,\, \alpha t^h+j
\bmod{\left(n-\sigma\right)}
\right)
\label{needequationnumberforasmallstep2}\\
&\stackrel{\text{(\ref{theoptimalpointwithrespecttopseudodistance})}}{\le}&
\mathop{\mathrm E}\left[\,
\chi\left[\,\sigma=1\,\right]\cdot d\left(\bs{u},n-1\right)
+
\sum_{j=0}^{n-\sigma-1}\,
\hat{d}\left(\bs{u},\, \bs{u}t^h+j
\bmod{\left(n-\sigma\right)}
\right)
\,\right]
\nonumber\\
&\le&
\mathop{\mathrm E}\left[\,
\chi\left[\,\sigma=1\,\right]\cdot \left(d\left(i',\bs{u}\right)+d\left(i',n-1\right)\right)
+
\sum_{j=0}^{n-\sigma-1}\,
\hat{d}\left(\bs{u},\,
\bs{u}t^h+j
\bmod{\left(n-\sigma\right)}
\right)
\,\right]
\nonumber\\
&=&
\chi\left[\,\sigma=1\,\right]\cdot
\left[
\left(
\frac{1}{n-\sigma}\cdot
\sum_{m=0}^{n-\sigma-1}\,
d\left(i',m\right)
\right)
+
d\left(i',n-1\right)
\right]
+
\mathop{\mathrm E}\left[\,
\sum_{j=0}^{n-\sigma-1}\,
\hat{d}\left(\bs{u},\,
\bs{u}t^h+j
\bmod{\left(n-\sigma\right)}
\right)
\,\right],\nonumber
\end{eqnarray}
}
where the last inequality (resp., equality) follows from the triangle
inequality for $d$ (resp., the uniform distribution of $\bs{u}$
over $\{0,1,\ldots,n-\sigma-1\}$).
Furthermore,
{\footnotesize 
\begin{eqnarray}
&&\mathop{\mathrm E}\left[\,
\sum_{j=0}^{n-\sigma-1}\,
\hat{d}\left(\bs{u},\,
\bs{u}t^h+j
\bmod{\left(n-\sigma\right)}
\right)
\,\right]\label{newlabel1on20180608}\\
&\stackrel{\text{(\ref{pseudodistance})}}{=}&
\mathop{\mathrm E}\left[\,
\sum_{j=0}^{n-\sigma-1}\,
\sum_{k=0}^{h-1}\,
d^{(n-\sigma)}
\left(\bs{u}t^k+\sum_{\ell=0}^{k-1}\, s_{h-1-\ell}(j)
\cdot t^{k-1-\ell},\,
\bs{u}t^{k+1}+\sum_{\ell=0}^{k}\, s_{h-1-\ell}(j)\cdot t^{k-\ell}\right)
\,\right]
\nonumber\\
&\le&
\mathop{\mathrm E}\left[\,
\sum_{j=0}^{n-\sigma-1}\,
\sum_{k=0}^{h-1}\, d^{(n-\sigma)}
\left(i',
\bs{u}t^k+\sum_{\ell=0}^{k-1}\, s_{h-1-\ell}(j)
\cdot t^{k-1-\ell}
\right)
+d^{(n-\sigma)}\left(i',
\bs{u}t^{k+1}+\sum_{\ell=0}^{k}\, s_{h-1-\ell}(j)\cdot t^{k-\ell}\right)
\,\right]
\nonumber\\
&\stackrel{\text{(\ref{themodularversionofdistances}) \& (\ref{theoptimalpointexcludingthelastone})}}{=}&
\sum_{j=0}^{n-\sigma-1}\,
\sum_{k=0}^{h-1}\,
\left(
\mathop{\mathrm E}\left[\,
d\left(i',\,
\bs{u}t^k+\sum_{\ell=0}^{k-1}\, s_{h-1-\ell}(j)
\cdot t^{k-1-\ell}
\bmod{\left(n-\sigma\right)}
\right)
\,\right]
\right.\nonumber\\
&&
\left.
+\mathop{\mathrm E}\left[\,d\left(i',\,
\bs{u}t^{k+1}+\sum_{\ell=0}^{k}\, s_{h-1-\ell}(j)\cdot t^{k-\ell}
\bmod{\left(n-\sigma\right)}
\right)
\,\right]
\right),
\label{newlabel2on20180608}
\end{eqnarray}
}
where the
inequality follows from the
triangle inequality for $d^{(n-\sigma)}$.
\comment{ 
By Eq.~(\ref{pseudodistance}),
\begin{eqnarray}
&&
\mathop{\mathrm E}\left[\,
\sum_{j=0}^{n-1}\,
\hat{d}\left(\bs{u},\bs{u}+j\bmod{n}
\right)\,\right]
\nonumber\\
&=&
\mathop{\mathrm E}\left[\,
\sum_{j=0}^{n-1}\,
\sum_{k=0}^{h-1}\, d\left(\bs{u}+\sum_{\ell=h-k}^{h-1}\, s_\ell(j)
\cdot t^\ell\bmod{n},
\bs{u}+\sum_{\ell=h-1-k}^{h-1}\, s_\ell(j)\cdot t^\ell\bmod{n}\right)
\,\right].
\nonumber
\end{eqnarray}
}

Because $\bs{u}$
distributes uniformly at random over
$\{0,1,\ldots,n-\sigma-1\}$ and
$\mathop{\mathrm{gcd}}(t,n-\sigma)=1$
by construction,
\begin{eqnarray*}
\bs{u}t^k + \sum_{\ell=0}^{k-1}\,
s_{h-1-\ell}(j)\cdot t^{k-1-\ell}
\bmod{\left(n-\sigma\right)},\\
\bs{u}t^{k+1} + \sum_{\ell=0}^{k}\,
s_{h-1-\ell}(j)\cdot t^{k-\ell}
\bmod{\left(n-\sigma\right)}
\end{eqnarray*}
are uniformly random elements of
$\{0,1,\ldots,n-\sigma-1\}$
for {\em any} $j\in\{0,1,\ldots,n-\sigma-1\}$ and
$k\in\{0,1,\ldots,h-1\}$.\footnote{Note that
$\mathop{\mathrm{gcd}}(t^k,n-\sigma)=\mathop{\mathrm{gcd}}(t^{k+1},n-\sigma)=1$.}
\footnote{In essence, this observation and
Eqs.~(\ref{newlabel1on20180608})--(\ref{newlabel2on20180608})
follow the intuitive exposition before Lemma~\ref{approximationratiolemma}.}
Therefore,
{\footnotesize 
\begin{eqnarray}
&&
\sum_{j=0}^{n-\sigma-1}\,
\sum_{k=0}^{h-1}\,
\left(
\mathop{\mathrm E}\left[\,
d\left(i',\,
\bs{u}t^k+\sum_{\ell=0}^{k-1}\, s_{h-1-\ell}(j)
\cdot t^{k-1-\ell}
\bmod{\left(n-\sigma\right)}
\right)
\,\right]
\right.\nonumber\\
&&
\left.
+\mathop{\mathrm E}\left[\,d\left(i',\,
\bs{u}t^{k+1}+\sum_{\ell=0}^{k}\, s_{h-1-\ell}(j)\cdot t^{k-\ell}
\bmod{\left(n-\sigma\right)}
\right)
\,\right]
\right)
\nonumber\\
\comment{ 
&&\sum_{j=0}^{n-\sigma-1}\,
\sum_{k=0}^{h-1}\,
\mathop{\mathrm E}\left[\,
d^{(n-\sigma)}\left(i',
\bs{u}t^k+\sum_{\ell=0}^{k-1}\, s_{h-1-\ell}(j)
\cdot t^{k-1-\ell}
\right)
\,\right]
+\mathop{\mathrm E}\left[\,d^{(n-\sigma)}\left(i',
\bs{u}t^{k+1}+\sum_{\ell=0}^{k}\, s_{h-1-\ell}(j)\cdot t^{k-\ell}\right)
\,\right]
\nonumber\\
}
&=&
\sum_{j=0}^{n-\sigma-1}\,
\sum_{k=0}^{h-1}\,
\left(
\frac{1}{n-\sigma}\cdot
\sum_{m=0}^{n-\sigma-1}\,
d\left(i',
m
\right)
+
\frac{1}{n-\sigma}\cdot
\sum_{m=0}^{n-\sigma-1}\,
d\left(i',
m
\right)
\right)\nonumber\\
&=&
2h
\sum_{m=0}^{n-\sigma-1}\,
d\left(i',m
\right).
\label{needequationnumberforasmallstep3}
\end{eqnarray}
}

Now,
{\footnotesize 
\begin{eqnarray}
&&
\sum_{j=0}^{n-1}\, d\left(\alpha,j\right)
\label{label20170701_16pm_a}\\
&\stackrel{\text{(\ref{needequationnumberforasmallstep1})--(\ref{needequationnumberforasmallstep2})}}{\le}&
\chi\left[\,\sigma=1\,\right]\cdot d\left(\alpha,n-1\right)
+
\sum_{j=0}^{n-\sigma-1}\,
\hat{d}\left(\alpha,\, \alpha t^h+j
\bmod{\left(n-\sigma\right)}
\right)
\label{label20170701_16pm_b}\\
&\stackrel{\text{(\ref{needequationnumberforasmallstep2})--(\ref{needequationnumberforasmallstep3})}}{\le}&
\chi\left[\,\sigma=1\,\right]\cdot
\left[
\left(
\frac{1}{n-\sigma}\cdot
\sum_{m=0}^{n-\sigma-1}\,
d\left(i',m\right)
\right)
+
d\left(i',n-1\right)
\right]
+
2h
\sum_{m=0}^{n-\sigma-1}\,
d\left(i',m
\right)\nonumber\\
&=&
\chi\left[\,\sigma=1\,\right]\cdot
\left[
\left(
\frac{1}{n-1}\cdot
\sum_{m=0}^{n-2}\,
d\left(i',m\right)
\right)
%
+
d\left(i',n-1\right)
\right]+
2h
\sum_{m=0}^{n-\sigma-1}\,
d\left(i',m
\right)\nonumber\\
&=&
\chi\left[\,\sigma=1\,\right]\cdot
\left[
\left(
\frac{1}{n-1}\cdot
\sum_{m=0}^{n-2}\,
d\left(i',m\right)
\right)
%
+
d\left(i',n-1\right)
\right]\nonumber\\
&+&
2h\cdot
\left(
\left(
\sum_{m=0}^{n-1}\,
d\left(i',m
\right)\right)
-\chi\left[\,\sigma=1\,\right]\cdot
d\left(i',n-1\right)
\right)\nonumber\\
&=&
2h
\cdot
\left(\sum_{m=0}^{n-1}\,d\left(i',m\right)\right)
-
\chi\left[\,\sigma=1\,\right]\cdot\left(
\left(2h-1\right)\cdot d\left(i',n-1\right)
-\frac{1}{n-1}\cdot \sum_{m=0}^{n-2}\,
d\left(i',m\right)
\right)\nonumber\\
&\stackrel{\text{(\ref{theoptimalpointexcludingthelastone})}}{=}&
2h
\cdot
\left(
\min_{i=0}^{n-\sigma-1}\,\sum_{j=0}^{n-1}\,d\left(i,j\right)
\right)
-
\chi\left[\,\sigma=1\,\right]\cdot\left(
\left(2h-1\right)\cdot d\left(i',n-1\right)
-\frac{1}{n-1}\cdot \sum_{m=0}^{n-2}\,
d\left(i',m\right)
\right).\,\,\,\,\,\,\,\,\,\,\,\,\label{label20170701_16pm_c}
\end{eqnarray}
}
Eqs.~(\ref{label20170701_16pm_b})~and~(\ref{label20170701_16pm_55min_a})
coincide.
Furthermore,
Eqs.~(\ref{label20170701_16pm_b})--(\ref{label20170701_16pm_c})
imply
Eq.~(\ref{label20170701_16pm_55min_b}).
\end{proof}

\begin{intuition}
When $\sigma=0$,
$\alpha$ in Eq.~(\ref{theoptimalpointwithrespecttopseudodistance})
is a $(2h)$-approximate $1$-median w.r.t.\ $d$ by
Lemma~\ref{approximationratiolemma}.
When $\sigma=1$, $\alpha\in\{0,1,\ldots,n-2\}$ by
Eq.~(\ref{theoptimalpointwithrespecttopseudodistance}),
missing $n-1$ from the
choice
of $\alpha$.
The next lemma
considers
$n-1$
as well.
In particular,
it shows how to pick a $(2h)$-approximate
$1$-median w.r.t.\ $d$
from $\{\alpha,n-1\}$.
In contrast,
Chang~\cite{Cha15, Cha15CMCT}
has no
$\sigma$ and, therefore,
always
includes $n-1$ in
picking
his
analogy of $\alpha$.
So
the next lemma
is independent of his works.
\end{intuition}
\comment{ 
The following lemma
shows how to pick a $(2h)$-approximate
$1$-median (w.r.t.\ $d$)
from $\{\alpha,n-1\}$.
}

\begin{lemma}
\label{selectingfromtwo}
Let $\alpha\in\{0,1,\ldots,n-\sigma-1\}$
be as in
Eq.~(\ref{theoptimalpointwithrespecttopseudodistance}),
breaking ties arbitrarily.
If
\begin{eqnarray}
\chi\left[\,\sigma=1\,\right]\cdot d\left(\alpha,n-1\right)
+\sum_{j=0}^{n-\sigma-1}\,
\hat{d}\left(\alpha,\,
\alpha t^h+j
\bmod{\left(n-\sigma\right)}
\right)
<\sum_{j=0}^{n-1}\, d\left(n-1,j\right),
\label{selectorequation}
\end{eqnarray}
then
\begin{eqnarray}
\sum_{j=0}^{n-1}\, d\left(\alpha,j\right)
\le 2h\cdot \min_{i=0}^{n-1}\,
\sum_{j=0}^{n-1}\,d\left(i,j\right).
\label{thenormalcase}
\end{eqnarray}
Otherwise,
\begin{eqnarray}
\sum_{j=0}^{n-1}\, d\left(n-1,j\right)
\le 2h\cdot \min_{i=0}^{n-1}\,
\sum_{j=0}^{n-1}\,d\left(i,j\right).
\label{thespecialcase}
\end{eqnarray}
\end{lemma}
\begin{proof}
Clearly,
\begin{eqnarray}
\sum_{j=0}^{n-1}\, d\left(\alpha,j\right)
&\stackrel{\text{(\ref{themodularversionofdistances})}}{=}&
\chi\left[\,\sigma=1\,\right]
\cdot d\left(\alpha,n-1\right)
+\sum_{j=0}^{n-\sigma-1}\,d^{(n-\sigma)}
\left(\alpha,j\right)
\label{justtrivialequation}\\
&\stackrel{\text{(\ref{pseudodistanceislarger})}}{\le}&
\chi\left[\,\sigma=1\,\right]
\cdot d\left(\alpha,n-1\right)
+\sum_{j=0}^{n-\sigma-1}\,\hat{d}\left(\alpha,j\right)
\nonumber\\
&=&
\chi\left[\,\sigma=1\,\right]
\cdot d\left(\alpha,n-1\right)
+\sum_{j=0}^{n-\sigma-1}\,
\hat{d}\left(\alpha,\,
\alpha t^h+j
\bmod{\left(n-\sigma\right)}
\right),\,\,\,\,\,\,\,\,\,\,\,\,\,
\label{repeatedusedexpansionthatisactuallysimple}
\end{eqnarray}
where the last equality uses the one-to-one
correspondence of
$j\mapsto \alpha t^h+j\bmod{(n-\sigma)}$
for $j\in\{0,1,\ldots,n-\sigma-1\}$.\footnote{Note that Eq.~(\ref{pseudodistanceislarger})
holds for {\em all}
$i$, $j\in\{0,1,\ldots,n-\sigma-1\}$, implying that
$d^{(n-\sigma)}(u,v)\le \hat{d}(u,v)$ for all $u$, $v\in\{0,1,\ldots,n-\sigma-1\}$.}

Next, we
separate the discussion as to whether
\begin{eqnarray}
\left(2h-1\right)\cdot d\left(i',n-1\right)
-\frac{1}{n-1}\cdot \sum_{m=0}^{n-2}\,
d\left(i',m\right)
\ge0.\label{theextradistances}
\end{eqnarray}
\begin{enumerate}[{Case}~(1):]
\item
Eq.~(\ref{theextradistances}) is true.
By
Lemma~\ref{approximationratiolemma},
\begin{eqnarray}
\sum_{j=0}^{n-1}\,d\left(\alpha,j\right)
\le
2h\cdot
\min_{i=0}^{n-\sigma-1}\,
\sum_{j=0}^{n-1}\, d\left(i,j\right).\,\,\,
\label{somederivations20160317_22pm}
\end{eqnarray}
\begin{enumerate}[{Subcase}~(i):]
\item Eq.~(\ref{selectorequation}) is true.
By
Eqs.~(\ref{selectorequation})~and~(\ref{justtrivialequation})--(\ref{repeatedusedexpansionthatisactuallysimple}),
\begin{eqnarray}
\sum_{j=0}^{n-1}\, d\left(\alpha,j\right)
<\sum_{j=0}^{n-1}\, d\left(n-1,j\right).
\nonumber
\end{eqnarray}
This and Eq.~(\ref{somederivations20160317_22pm})
imply Eq.~(\ref{thenormalcase}).
\item Eq.~(\ref{selectorequation}) is false.
By
Eq.~(\ref{label20170701_16pm_55min_b})
of
Lemma~\ref{approximationratiolemma} and
Eq.~(\ref{theextradistances}),
{\small 
\begin{eqnarray}
&&
\chi\left[\,\sigma=1\,\right]\cdot d\left(\alpha,n-1\right)
+
\sum_{j=0}^{n-\sigma-1}\,
\hat{d}\left(\alpha,\,\alpha t^h+j
\bmod{\left(n-\sigma\right)}
\right)
\nonumber\\
&\le& 2h
\cdot
\min_{i=0}^{n-\sigma-1}\,
\sum_{j=0}^{n-1}\,d\left(i,j\right).
\nonumber
\end{eqnarray}
}
This and
the negation of
Eq.~(\ref{selectorequation})
imply
\begin{eqnarray}
\sum_{j=0}^{n-1}\,d\left(n-1,j\right)
\le 2h
\cdot
\min_{i=0}^{n-\sigma-1}\,
\sum_{j=0}^{n-1}\,d\left(i,j\right).
\label{thelastpointisbetterthanpreviousones}
\end{eqnarray}
Eq.~(\ref{thelastpointisbetterthanpreviousones})
implies
Eq.~(\ref{thespecialcase}) (note that
$\sum_{j=0}^{n-1}\, d(n-1,j)$ does not exceed itself).
\end{enumerate}
\item
Eq.~(\ref{theextradistances}) is false.
By the triangle inequality for $d$,
\begin{eqnarray}
\sum_{j=0}^{n-1}\, d\left(n-1,j\right)
&\le&
\sum_{j=0}^{n-1}\,
\left(d\left(i',n-1\right)+d\left(i',j\right)\right)
\label{newlabel20170702a}\\
&=&n\cdot d\left(i',n-1\right)
+\sum_{j=0}^{n-1}\,d\left(i',j\right)\label{newlabel20170702b}.
\end{eqnarray}
Eqs.~(\ref{newlabel20170702a})--(\ref{newlabel20170702b}) and
the negation of
Eq.~(\ref{theextradistances})
imply
\begin{eqnarray}
\sum_{j=0}^{n-1}\, d\left(n-1,j\right)
&<&
n\cdot\frac{1}{2h-1}\cdot \frac{1}{n-1}
\cdot\left(\sum_{m=0}^{n-2}\,d\left(i',m\right)\right)
+\sum_{j=0}^{n-1}\,d\left(i',j\right)
\,\,\,\,\,\,\,\,\,\,\,
\label{stillneedanequationnumberhere}\\
&\le& 2\cdot\sum_{j=0}^{n-1}\,d\left(i',j\right)
\nonumber\\
&\stackrel{\text{(\ref{theoptimalpointexcludingthelastone})}}{=}&2\cdot \min_{i=0}^{n-\sigma-1}\,\sum_{j=0}^{n-1}\,d\left(i,j\right),
\label{stillneedanequationnumberhereagain}
\end{eqnarray}
where the second inequality uses $h\ge 2$.
Eqs.~(\ref{stillneedanequationnumberhere})--(\ref{stillneedanequationnumberhereagain})
imply
Eq.~(\ref{thespecialcase}).
\begin{enumerate}[{Subcase}~(a):]
\item Eq.~(\ref{selectorequation}) is false.
We have proved
Eq.~(\ref{thespecialcase}),
as desired.
\item Eq.~(\ref{selectorequation}) is true.
By
Eqs.~(\ref{selectorequation})~and~(\ref{justtrivialequation})--(\ref{repeatedusedexpansionthatisactuallysimple}),
$$
\sum_{j=0}^{n-1}\, d\left(\alpha,j\right)
<\sum_{j=0}^{n-1}\, d\left(n-1,j\right).
$$
This and Eq.~(\ref{thespecialcase}) (which has been proved)
give Eq.~(\ref{thenormalcase}).
\end{enumerate}
\end{enumerate}
\end{proof}

\comment{ 
When $\sigma=1$, $\alpha\in\{0,1,\ldots,n-2\}$ by
Eq.~(\ref{theoptimalpointwithrespecttopseudodistance}).
In contrast,
Chang~\cite{Cha15}
has no
$\sigma$ and, therefore,
always
includes $n-1$ in
picking
his
analogy of $\alpha$.
So
Lemma~\ref{selectingfromtwo}, which considers $n-1$ separately from $\alpha$,
is independent of his work.
}


\comment{ 
\begin{figure}
\begin{algorithmic}[1]
\FOR{$i\in\{0,1,\ldots,n-1\}$}
  \FOR{$j\in\{0,1,\ldots,n-1\}$}
    \STATE Compute $\hat{d}(i,i+j\bmod{n})$ according to Eq.~(\ref{pseudodistance});
  \ENDFOR
\ENDFOR
\STATE Output $\mathop{\mathrm{argmin}}_{i=0}^{n-1}\,\sum_{j=0}^{n-1}\,\hat{d}(i,i+j\bmod{n})$, breaking ties arbitrarily;
\end{algorithmic}
\caption{Algorithm {\sf simple-median} with input a metric space $(\{0,1,\ldots,n-1\},d)$
and $h\in\mathbb{Z}^+\setminus\{1\}$}
\label{simplealgorithm}
\end{figure}
}

\comment{ 
By Lemma~\ref{approximationratiolemma},
algorithm {\sf simple-median} in Fig.~\ref{simplealgorithm}
is $(2h)$-approximate for {\sc metric $1$-median}.
By Eq.~(\ref{pseudodistance}),
$\hat{d}(\cdot,\cdot)$ depends on $d(p,q)$ only if
$$q-p\equiv s\cdot t^{h-1-k}\pmod{n}$$
for some
$s\in\{0,1,\ldots,t-1\}$
and $k\in\{0,1,\ldots,h-1\}$.
So {\sf simple-median} makes no more than $nht=O(hn^{1+1/h})$ distinct queries to $d$.
Furthermore, it is clearly deterministic and nonadaptive.
Therefore, we
are left
only to improve its running time to $O(hn^{1+1/h})$.
}

\section{Dynamic programming}

By
Lemma~\ref{selectingfromtwo},
one of $\alpha$
and $n-1$
is
a $(2h)$-approximate $1$-median.
This section
finds $\alpha$ by dynamic programming.
Details follow.

Define
$(s'_{h-1},s'_{h-2},\ldots,s'_0)\in\{0,1,\ldots,t-1\}^h$
to be
the $t$-ary representation of $n-\sigma-1$.
So $\sum_{r=0}^{h-1}\,s'_r\cdot t^r=n-\sigma-1$.
For $i\in\{0,1,\ldots,n-\sigma-1\}$
and $m\in\{0,1,\ldots,h-1\}$,
define
{\small 
\begin{eqnarray}
f\left(i,m\right)
&\equiv&
\sum_{s_m,s_{m-1},\ldots,s_0=0}^{t-1}\,
\chi\left[\sum_{r=0}^m\, s_r\cdot t^r\le \sum_{r=0}^m\, s'_r\cdot t^r
\right]\nonumber\\
&\cdot&\sum_{k=0}^m\,
d^{(n-\sigma)}
\left(it^k+\sum_{\ell=0}^{k-1}\, s_{m-\ell}\cdot t^{k-1-\ell},\,
it^{k+1}+\sum_{\ell=0}^k\, s_{m-\ell}\cdot t^{k-\ell}
\right),
\label{subsumlessthanorequalto}\\
g\left(i,m\right)
&\equiv&
\sum_{s_m,s_{m-1},\ldots,s_0=0}^{t-1}\,
\sum_{k=0}^m\,
d^{(n-\sigma)}
\left(it^k+\sum_{\ell=0}^{k-1}\, s_{m-\ell}\cdot t^{k-1-\ell},\,
it^{k+1}+\sum_{\ell=0}^k\, s_{m-\ell}\cdot t^{k-\ell}
\right);
\,\,\,\,\,\,\,\,\,\,\,\,\,\,\,\label{subsumlessthan}
\end{eqnarray}
}
hence
\begin{eqnarray}
f\left(i,0\right)
&=&
\sum_{s_0=0}^{s'_0}\,
d^{(n-\sigma)}\left(i,\, it+s_0\right),
\label{thebasecaseofthefirstfunction}\\
g\left(i,0\right)
&=&
\sum_{s_0=0}^{t-1}\,
d^{(n-\sigma)}\left(i,\, it+s_0\right).
\label{thebasecaseofthesecondfunction}
\end{eqnarray}
Chang
also defines
functions similar to our $f(\cdot,\cdot)$ and
$g(\cdot,\cdot)$~\cite[Eqs.~(8)--(9)]{Cha15}, based on his
pseudo-distance function~\cite[Eq.~(2)]{Cha15}.
Instead, Eqs.~(\ref{subsumlessthanorequalto})--(\ref{subsumlessthan})
are based on $\hat{d}$ in
Eq.~(\ref{pseudodistance}).

\comment{ 
The following lemma
implies that
line~6 of {\sf simple-median}
outputs
$$\mathop{\mathrm{argmin}}_{i=0}^{n-1}\,f\left(i,h-1\right),$$
breaking ties arbitrarily.
}

The
following lemma
implies
$f(i,h-1)=\sum_{j=0}^{n-\sigma-1}\,\hat{d}(i,j)$
for all $i\in\{0,1,\ldots,n-\sigma-1\}$:

\begin{lemma}[{cf.~\cite[Lemma~5]{Cha15}}]
\label{theDPresultisthesumofpseudodistanceslemma}
For all $i\in\{0,1,\ldots,n-\sigma-1\}$,
\begin{eqnarray}
f\left(i,h-1\right)
=\sum_{j=0}^{n-\sigma-1}\,
\hat{d}\left(i,\, it^h+j
\bmod{\left(n-\sigma\right)}
\right).
\label{newlabel20170702_21pm}
\end{eqnarray}
\end{lemma}
\begin{proof}
Representing each $j\in\{0,1,\ldots,n-\sigma-1\}$
in $t$-ary as $(s_{h-1},s_{h-2},\ldots,s_0)$,
\begin{eqnarray}
&&\sum_{j=0}^{n-\sigma-1}\,
\sum_{k=0}^{h-1}\, d^{(n-\sigma)}
\left(it^k+\sum_{\ell=0}^{k-1}\,s_{h-1-\ell}(j)\cdot t^{k-1-\ell},
\,
it^{k+1}+\sum_{\ell=0}^k\,s_{h-1-\ell}(j)\cdot t^{k-\ell}\right)
\nonumber\\
&=&
\sum_{s_{h-1},s_{h-2},\ldots,s_0=0}^{t-1}\,
\chi\left[\sum_{r=0}^{h-1}\, s_r\cdot t^r\le
n-\sigma-1
\right]\nonumber\\
&\cdot&\sum_{k=0}^{h-1}\,
d^{(n-\sigma)}
\left(it^k+\sum_{\ell=0}^{k-1}\,
s_{h-1-\ell}\cdot t^{k-1-\ell},\,
it^{k+1}+\sum_{\ell=0}^k\, s_{h-1-\ell}\cdot t^{k-\ell}
\right).
\label{justaneasyequation}
\end{eqnarray}
Eqs.~(\ref{pseudodistance}),~(\ref{subsumlessthanorequalto})~and~(\ref{justaneasyequation})
complete the proof (recall that
$\sum_{r=0}^{h-1}\,s'_r\cdot t^r= n-\sigma-1$).\footnote{By Eq.~(\ref{pseudodistance}),
the right-hand side of Eq.~(\ref{newlabel20170702_21pm}) coincides with the left-hand
side of Eq.~(\ref{justaneasyequation}).
By Eq.~(\ref{subsumlessthanorequalto}) and recalling that
$\sum_{r=0}^{h-1}\,s'_r\cdot t^r= n-\sigma-1$, the right-hand side of
Eq.~(\ref{justaneasyequation}) is $f(i,h-1)$.}
\end{proof}


When $\sigma=0$,
a minimizer of $f(\cdot,h-1)$
is a $(2h)$-approximate $1$-median w.r.t.\ $d$
by
Lemmas~\ref{approximationratiolemma}~and~\ref{theDPresultisthesumofpseudodistanceslemma}.
So we want to
calculate $f(i,h-1)$ for all $i\in\{0,1,\ldots,n-\sigma-1\}$.
Similar to~\cite{Cha15}, we do so by dynamic programming.
For this purpose,
we need the following
recurrences for $g(\cdot,\cdot)$ and $f(\cdot,\cdot)$, whose
very
cumbersome
(but na{\"\i}ve)
proofs are
in
Appendices~\ref{appendix1}--\ref{appendix2},
respectively:

\comment{ 
represent $g(\cdot,m)$
and $f(\cdot,m)$ in terms of $g(\cdot,m-1)$ and $f(\cdot,m-1)$
for $m\in\{1,2,\ldots,h-1\}$,
which together with
Eqs.~(\ref{thebasecaseofthefirstfunction})--(\ref{thebasecaseofthesecondfunction})
allows us to compute
$f(\cdot,0)$, $f(\cdot,1)$, $\ldots$, $f(\cdot,h-1)$, in that order,
by dynamic programming.
}



\begin{lemma}[{cf.~\cite[Lemma~6]{Cha15}}]
\label{recurrencelemma1}
For all $i\in\{0,1,\ldots,n-\sigma-1\}$
and $m\in\{1,2,\ldots,h-1\}$,
\begin{eqnarray}
g\left(i,m
\right)
&=&
t^m\sum_{s_m=0}^{t-1}\,
d^{(n-\sigma)}\left(i,\, it+s_m\right)\nonumber\\
&+&\sum_{s_m=0}^{t-1}\,
g\left(it+s_m
\bmod{\left(n-\sigma\right)},\, m-1
\right).
\nonumber
\end{eqnarray}
\end{lemma}


\comment{ 
The
recurrence for $f(\cdot,\cdot)$
has a tedious proof
without
special techniques, as follows.
}

\begin{lemma}[{cf.~\cite[Lemma~7]{Cha15}}]
\label{recurrencelemma2}
For all $i\in\{0,1,\ldots,n-\sigma-1\}$ and $m\in\{1,2,\ldots,h-1\}$,
\begin{eqnarray}
f\left(i,m\right)
&=&
\left(1+\sum_{r=0}^{m-1}\, s'_r\cdot t^r\right)
d^{(n-\sigma)}
\left(i,it+s'_m\right)\nonumber\\
&+&
t^m
\sum_{s_m=0}^{s'_m-1}\,
d^{(n-\sigma)}
\left(i,it+s_m\right)\nonumber\\
&+&
f\left(it+s'_m
\bmod{\left(n-\sigma\right)},\,
m-1\right)\nonumber\\
&+&\sum_{s_m=0}^{s'_m-1}\, g\left(it+s_m
\bmod{\left(n-\sigma\right)},\,
m-1
\right).
\nonumber
\end{eqnarray}
\end{lemma}

\comment{ 
We move the
cumbersome
proofs of Lemmas~\ref{recurrencelemma1}--\ref{recurrencelemma2}
to Appendices~\ref{appendix1}--\ref{appendix2}, respectively.
}

\begin{figure}
\begin{algorithmic}[1]
\STATE
Pick any prime number $t\in[\,\lceil n^{1/h}\rceil,\,2\cdot \lceil n^{1/h}\rceil\,]$;
\STATE Pick any $\sigma\in\{0,1\}$ satisfying $\mathop{\mathrm{gcd}}
(t,n-\sigma)=1$;
\FOR{$i=0$, $1$, $\ldots$, $n-\sigma-1$}
  \FOR{$s=0$, $1$, $\ldots$, $t-1$}
    \STATE Query for $d(i,it+s\bmod{(n-\sigma)})$;
  \ENDFOR
  \STATE Query for $d(n-1,i)$;
\ENDFOR
\STATE
$(s'_{h-1},s'_{h-2},\ldots,s'_0)
\leftarrow\text{the $t$-ary representation of $n-\sigma-1$}$;
\FOR{$i=0$, $1$, $\ldots$, $n-\sigma-1$}
  \STATE $f[i][0]\leftarrow \sum_{s_0=0}^{s'_0}\, d(i,it+s_0\bmod{(n-\sigma)})$;
  \STATE $g[i][0]\leftarrow \sum_{s_0=0}^{t-1}\, d(i,it+s_0\bmod{(n-\sigma)})$;
\ENDFOR
\FOR{$m=1$, $2$, $\ldots$, $h-1$}
  \FOR{$i=0$, $1$, $\ldots$, $n-\sigma-1$}
    \STATE $g[i][m]\leftarrow t^m\sum_{s_m=0}^{t-1}\,
d(i,it+s_m\bmod{(n-\sigma)})$;
    \STATE $g[i][m]\leftarrow g[i][m]+\sum_{s_m=0}^{t-1}\,
g[it+s_m\bmod{(n-\sigma)}][m-1]$;
    \STATE $f[i][m]\leftarrow (1+\sum_{r=0}^{m-1}\,s'_r\cdot t^r)\,
d(i,it+s'_m\bmod{(n-\sigma)})$;
    \STATE $f[i][m]\leftarrow f[i][m]+t^m\sum_{s_m=0}^{s'_m-1}\,
d(i,it+s_m\bmod{(n-\sigma)})$;
    \STATE $f[i][m]\leftarrow f[i][m]+f[it+s'_m\bmod{(n-\sigma)}][m-1]$;
    \STATE $f[i][m]\leftarrow f[i][m]+\sum_{s_m=0}^{s'_m-1}\,
g[it+s_m\bmod{(n-\sigma)}][m-1]$;
  \ENDFOR
\ENDFOR
\STATE
$\alpha\leftarrow \mathop{\mathrm{argmin}}_{i=0}^{n-\sigma-1}\,
(\chi[\sigma=1]\cdot d(i,n-1)+f[i][h-1])$,
breaking ties arbitrarily;
\IF{$\chi[\sigma=1]\cdot d(\alpha,n-1)+f[\alpha][h-1]
<\sum_{j=0}^{n-1}\,d(n-1,j)$}
  \STATE Output $\alpha$;
\ELSE
  \STATE Output $n-1$;
\ENDIF
\end{algorithmic}
\caption{Algorithm {\sf Approx.-Median}}
\label{mainalgorithm}
\end{figure}

\comment{ 
Lemmas~\ref{recurrencelemma1}--\ref{recurrencelemma2}
allow us to compute
$f(\cdot)$
by dynamic programming.
}

\comment{ 
\begin{lemma}\label{theoutputminimizessomething}
Line~17 of
{\sf Approx.-Median} in Fig.~\ref{mainalgorithm}
outputs $\mathop{\mathrm{argmin}}_{i=0}^{n-1}\,
f(i,h-1)$, breaking ties arbitrarily.
\end{lemma}
\begin{proof}
By
Eqs.~(\ref{thebasecaseofthefirstfunction})--(\ref{thebasecaseofthesecondfunction}),
lines~3--6 of {\sf Approx.-Median}
find $f(\cdot,0)$ and $g(\cdot,0)$.
By
Lemmas~\ref{recurrencelemma1}~and~\ref{recurrencelemma2},
lines~8--15
derive
$f(\cdot,m)$ and $g(\cdot,m)$
from $f(\cdot,m-1)$ and $g(\cdot,m-1)$.
So lines~7--15 eventually find $f(\cdot,h-1)$ and $g(\cdot,h-1)$.
\end{proof}
}

Lemmas~\ref{recurrencelemma1}--\ref{recurrencelemma2} express
$g(\cdot,m)$ and $f(\cdot,m)$ using $g(\cdot,m-1)$, $f(\cdot,m-1)$
and $d^{(n-\sigma)}(\cdot,\cdot)$.
So we
can find
$f(\cdot,0)$, $f(\cdot,1)$, $\ldots$,
$f(\cdot,h-1)$, in that order.
A
$(2h)$-approximate $1$-median can thus be found:

\begin{lemma}[{cf.~\cite[Lemma~8]{Cha15}}]
\label{algorithmapproximationratiolemma}
{\sf Approx.-Median} in Fig.~\ref{mainalgorithm}
outputs a $(2h)$-approximate $1$-median w.r.t.\ $d$.
\end{lemma}
\begin{proof}
By
Eqs.~(\ref{thebasecaseofthefirstfunction})--(\ref{thebasecaseofthesecondfunction}),
lines~10--13 of {\sf Approx.-Median}
find $f(\cdot,0)$ and $g(\cdot,0)$.
By
Lemmas~\ref{recurrencelemma1}~and~\ref{recurrencelemma2},
lines~14--23
find
$f(\cdot,m)$ and $g(\cdot,m)$
for an increasing $m\in\{1,2,\ldots,h-1\}$.
%
By Lemma~\ref{theDPresultisthesumofpseudodistanceslemma},
line~24 picks $\alpha$ as in
Eq.~(\ref{theoptimalpointwithrespecttopseudodistance}).
Also by Lemma~\ref{theDPresultisthesumofpseudodistanceslemma},
the condition in line~25 is the same as Eq.~(\ref{selectorequation}).
So by
Lemma~\ref{selectingfromtwo},
lines~25--29
output
a
$(2h)$-approximate
$1$-median w.r.t.\ $d$.

By line~9, $0\le s'_m\le t-1$ for all $m\in\{0,1,\ldots,h-1\}$.
So all $d$-distances in lines~11--12,~16~and~18--19
are queried for in line~5 of the loops in lines~3--8.
Instead,
the $d$-distances in lines~24--25 are
queried for in
line~7.
\end{proof}

Below is
our main theorem.

\begin{theorem}[{cf.~\cite[Theorem~9]{Cha15}}]
\label{maintheorem}
Let
$h\colon\mathbb{Z}^+\to\mathbb{Z}^+\setminus\{1\}$
be
any function
such that
$h(n)$ and
$\lceil n^{1/h(n)}\rceil$
are computable from $n$ in $O(h(n)\cdot n^{1+1/h(n)})$ time.
Then
{\sc metric $1$-median} has a
deterministic,
$O(h(n)\cdot n^{1+1/h(n)})$-time,
$O(n^{1+1/h(n)})$-query,
$(2\,h(n))$-approximation and
nonadaptive
algorithm.
\end{theorem}
\begin{proof}
By
Lemma~\ref{algorithmapproximationratiolemma},
{\sf Approx.-Median} is $(2h)$-approximate.
It is clearly
deterministic and nonadaptive.\footnote{Clearly, lines~3--8
of {\sf Approx.-Median}
make nonadaptive queries.}
By the
well-known AKS primality test,
line~1 of {\sf Approx.-Median}
takes
time
$$\left(\left\lceil n^{1/h}\right\rceil+1\right)
\cdot\log^{O(1)}\left(O\left(\lceil n^{1/h}\rceil\right)\right)
=\tilde{O}\left(\sqrt{n}\right),$$
where the
equality uses $h\ge2$.
Because
$\mathop{\mathrm{gcd}}(\cdot,\cdot)$
is polynomial-time computable and $t=O(n)$ by line~1,
line~2 takes $\log^{O(1)} n$ time.\footnote{In general,
$\mathop{\mathrm{gcd}}(\cdot,\cdot)$ takes time polynomial
in the lengths (not the values) of the binary representations of its inputs.}
Lines~3--8
make $O(nt)$ queries in $O(nt)$ time.
As $0\le s'_0\le t-1$ by line~9,
lines~10--13 take $O(nt)$ time.
By precomputing
$\{t^i\}_{i=0}^{h-1}$ and
$\{\sum_{r=0}^i\, s'_r\cdot t^r\}_{i=0}^{h-1}$
in $O(h)$
time,\footnote{Calculate $t^i=t\cdot t^{i-1}$ in the increasing order of $i\in\{1,2,\ldots,h-1\}$.
Then calculate $\sum_{r=0}^i\, s'_r\cdot t^r=s'_i\cdot t^i+\sum_{r=0}^{i-1}\, s'_r\cdot t^r$
in the increasing order of $i\in\{0,1,\ldots,h-1\}$.
Because $t=\Theta(n^{1/h})$ and $\sum_{r=0}^{h-1}\, s'_r\cdot t^r=n-\sigma-1$,
the numbers appearing in
calculating $\{t^i\}_{i=0}^{h-1}$ and
$\{\sum_{r=0}^i\, s'_r\cdot t^r\}_{i=0}^{h-1}$
are
$O(h+\lg n)$- and $O(\lg n)$-bit
long, respectively.
As $n^{1/h}=\Theta(1)$ for all $h\ge\lg n$, the theorem is stronger with $h=\lg n$ than
with $h>\lg n$ (because the desired time complexity, query complexity and approximation ratio
of $O(hn^{1+1/h})$, $O(n^{1+1/h})$ and $2h$ are $O(hn)$,
$O(n)$ and $2h$, respectively, for all $h\ge\lg n$, and the theorem is stronger with
a smaller time complexity and a smaller approximation ratio).
So we may assume $h\le \lg n$ in the proof.
Consequently, the numbers appearing in
calculating $\{t^i\}_{i=0}^{h-1}$ and
$\{\sum_{r=0}^i\, s'_r\cdot t^r\}_{i=0}^{h-1}$ are $O(\lg n)$-bit long.
That arithmetic over $O(\lg n)$-bit numbers take $O(1)$ time is standard under most
algorithmic models.}
lines~14--23 take $O(hnt)$ time.\footnote{Because $0\le s'_m\le t-1$ for
all $m\in\{0,1,\ldots,h-1\}$ by line~9 and $\{t^i\}_{i=0}^{h-1}$ has been precomputed,
each execution
of
lines~16--17,~19,~and~21
takes $O(t)$ time.
Instead, an execution of line~18 takes $O(1)$ time
because $\sum_{r=0}^{m-1}\,s'_r\cdot t^r$ has been precomputed.}
As
$t=\Theta(n^{1/h})$ by line~1,
the overall time (resp., query) complexity
of $O(hnt)$ (resp., $O(nt)$) is as desired.
\end{proof}

Note that the condition on $h$ in Theorem~\ref{maintheorem}
is mild; it holds for most commonly seen
functions.
For example, we have the following corollary:


\begin{corollary}\label{maincorollary}
{\sc Metric $1$-median} has a deterministic, $O(\exp(O(1/\epsilon))\cdot n\log n)$-time,
$O(\exp(O(1/\epsilon))\cdot n)$-query, $(\epsilon\log n)$-approximation and nonadaptive algorithm
for each constant $\epsilon>0$.
\end{corollary}
\begin{proof}
Take $h(n)=(\epsilon/2) \lg n$
in Theorem~\ref{maintheorem}.
\end{proof}

Corollary~\ref{maincorollary}
is stronger than taking $h(n)=(\epsilon/2) \lg n$ in
Theorem~\ref{oldresultgeneralized}.
But it leaves open
whether deterministic
$O(n\log n)$-time or
$O(n)$-query algorithms
for {\sc metric $1$-median} can be
$o(\log n)$-approximate.

Unlike
in line~5 of {\sf Approx.-Median},
Chang's~\cite{Cha15}
pseudo-distance function
requires to query for
$d(i,i+st^k\bmod{n})$ for all $i\in\{0,1,\ldots,n-1\}$,
$s\in\{0,1,\ldots,t-1\}$ and $k\in\{0,1,\ldots,h-1\}$ (albeit with
a slightly different $t$).
So his query complexity is $\Theta(hnt)=\Theta(h(n)\cdot n^{1+1/h(n)})$,
as compared to our $O(n^{1+1/h(n)})$
in Theorem~\ref{maintheorem}.

\comment{ 
\section{Conclusions}

Theorem~\ref{maintheorem} improves
the query complexity in
Theorem~\ref{oldresultgeneralized} to $O(n^{1+1/h(n)})$.
Future work remains to close the gap on the query complexities between
our result and
the following lower
bound:
\begin{theorem}[\cite{Cha18}]\label{currentlybestlowerbound}
{\sc Metric $1$-median} has no deterministic $o(n^{1+1/(h(n)-1)}/h(n))$-query
$(2h(n)\cdot(1-\epsilon))$-approximation algorithms for any constant $\epsilon>0$
and any $h\colon \mathbb{Z}^+\to
\mathbb{Z}^+\setminus\{1\}$ satisfying $h=o(n^{1/(h(n)-1)})$.
\end{theorem}
}

\section*{Acknowledgments}

The author is supported in part by the Ministry of Science and Technology of
Taiwan under grant 106-2221-E-155-013-.


\appendix
\section{Proof of Lemma~\ref{recurrencelemma1}}\label{appendix1}

We begin with the following change-of-variable formula:

\begin{lemma}[{cf.~\cite[Lemma~6]{Cha15}}]
\label{tedioustwicechangeofvariables}
For all $m\in\{0,1,\ldots,h-1\}$ and $s_m$, $s_{m-1}$, $\ldots$,
$s_0\in\{0,1,\ldots,t-1\}$,
{\small 
\begin{eqnarray*}
&&\sum_{k=1}^{m}\,
d^{(n-\sigma)}
\left(it^k+\sum_{\ell=0}^{k-1}\, s_{m-\ell}\cdot t^{k-1-\ell},
\,
it^{k+1}+\sum_{\ell=0}^{k}\,s_{m-\ell}\cdot t^{k-\ell}
\right)
\nonumber\\
&=&
\sum_{k=0}^{m-1}\,
d^{(n-\sigma)}
\left(it^{k+1}+s_m\cdot t^k+\sum_{\ell=0}^{k-1}\, s_{m-1-\ell}\cdot t^{k-1-\ell},\,
it^{k+2}+s_m\cdot t^{k+1}+\sum_{\ell=0}^{k}\,s_{m-1-\ell}\cdot t^{k-\ell}
\right)
\end{eqnarray*}
}
\end{lemma}
\begin{proof}
Clearly,
{\footnotesize 
\begin{eqnarray*}
&&\sum_{k=1}^{m}\,
d^{(n-\sigma)}
\left(it^k+\sum_{\ell=0}^{k-1}\, s_{m-\ell}\cdot t^{k-1-\ell},\,
it^{k+1}+\sum_{\ell=0}^{k}\,s_{m-\ell}\cdot t^{k-\ell}
\right)
\nonumber\\
&=&
\sum_{k=0}^{m-1}\,
d^{(n-\sigma)}
\left(it^{k+1}+\sum_{\ell=0}^{k}\, s_{m-\ell}\cdot t^{k-\ell},\,
it^{k+2}+\sum_{\ell=0}^{k+1}\,s_{m-\ell}\cdot t^{k+1-\ell}
\right)
\nonumber\\
&=&
\sum_{k=0}^{m-1}\,
d^{(n-\sigma)}
\left(it^{k+1}+s_m\cdot t^k+\sum_{\ell=1}^{k}\, s_{m-\ell}\cdot t^{k-\ell},\,
it^{k+2}+s_m\cdot t^{k+1}+\sum_{\ell=1}^{k+1}\,s_{m-\ell}\cdot t^{k+1-\ell}
\right)
\nonumber\\
&=&
\sum_{k=0}^{m-1}\,
d^{(n-\sigma)}
\left(it^{k+1}+s_m\cdot t^k+\sum_{\ell=0}^{k-1}\, s_{m-1-\ell}\cdot t^{k-1-\ell},\,
it^{k+2}+s_m\cdot t^{k+1}+\sum_{\ell=0}^{k}\,s_{m-1-\ell}\cdot t^{k-\ell}
\right),
\end{eqnarray*}
}
where
the first and the last equalities follow from
substituting $k$ with $k+1$ and $\ell$ with $\ell+1$, respectively.
\end{proof}

\begin{restate1}[{cf.~\cite[Lemma~6]{Cha15}}]
For all $i\in\{0,1,\ldots,n-\sigma-1\}$
and $m\in\{1,2,\ldots,h-1\}$,
\begin{eqnarray}
g\left(i,m
\right)
&=&
t^m\sum_{s_m=0}^{t-1}\,
d^{(n-\sigma)}\left(i,\, it+s_m\right)\nonumber\\
&+&\sum_{s_m=0}^{t-1}\,
g\left(it+s_m
\bmod{\left(n-\sigma\right)},\, m-1
\right).
\nonumber
\end{eqnarray}
\end{restate1}
\begin{proof}
By Eq.~(\ref{subsumlessthan}),
{\footnotesize 
\begin{eqnarray}
&&g\left(i,m\right)\label{simplerrecurrencestart}\\
&=&
\sum_{s_m=0}^{t-1}\,
\sum_{s_{m-1},s_{m-2},\ldots,s_0=0}^{t-1}\,
\left[\,d^{(n-\sigma)}\left(i,it+s_m
\right)
\vphantom{\sum_{\ell=m-1-k}^{m-2}\,s_\ell\cdot t^\ell}
\right.\nonumber\\
&&\left.
\right.\nonumber\\
&&
\left.
+\sum_{k=1}^{m}\,
d^{(n-\sigma)}
\left(it^k+\sum_{\ell=0}^{k-1}\, s_{m-\ell}\cdot t^{k-1-\ell},\,
it^{k+1}+\sum_{\ell=0}^{k}\,s_{m-\ell}\cdot t^{k-\ell}
\right)
\,\right]\nonumber\\
&=&
\sum_{s_m=0}^{t-1}\,
t^m\cdot
d^{(n-\sigma)}\left(i,it+s_m
\right)\nonumber\\
&+&
\sum_{s_m=0}^{t-1}\,
\sum_{s_{m-1},s_{m-2},\ldots,s_0=0}^{t-1}\,
\sum_{k=1}^{m}\,
d^{(n-\sigma)}
\left(it^k+\sum_{\ell=0}^{k-1}\, s_{m-\ell}\cdot t^{k-1-\ell},\,
it^{k+1}+\sum_{\ell=0}^{k}\,s_{m-\ell}\cdot t^{k-\ell}
\right),
\nonumber
\end{eqnarray}
}
where the last equality holds because
there are $t^m$ tuples $(s_{m-1},s_{m-2},\ldots,s_0)\in\{0,1,\ldots,t-1\}^m$.
Again by Eq.~(\ref{subsumlessthan}),
{\footnotesize 
\begin{eqnarray}
&&g\left(it+s_m
\bmod{\left(n-\sigma\right)},\,
m-1\right)\nonumber\\
&=&
\sum_{s_{m-1},s_{m-2},\ldots,s_0=0}^{t-1}\,
\sum_{k=0}^{m-1}\,
d^{(n-\sigma)}
\left(\left(it+s_m\right)t^k+\sum_{\ell=0}^{k-1}\,s_{m-1-\ell}\cdot t^{k-1-\ell},\,
\left(it+s_m\right)t^{k+1}+\sum_{\ell=0}^{k}\,s_{m-1-\ell}\cdot t^{k-\ell}
\right)
\,\,\,\,\,\,\,\,\,\,\,
\label{simplerrecurrencepart2}
\end{eqnarray}
}
for $s_m\in\{0,1,\ldots,t-1\}$,
where
we use
Eq.~(\ref{themodularversionofdistances}) as
well.
Eqs.~(\ref{simplerrecurrencestart})--(\ref{simplerrecurrencepart2})
and
Lemma~\ref{tedioustwicechangeofvariables}
complete the proof.
\comment{ 
it remains
to verify the easy fact that
\begin{eqnarray}
\sum_{s_m=0}^{t-1}\,
\sum_{s_{m-1},s_{m-2},\ldots,s_0=0}^{t-1}\,
d\left(i,it+s_m\right)
=
t^m\sum_{s_m=0}^{t-1}\,
d\left(it,s_m
\right).
\nonumber
\end{eqnarray}
}
\end{proof}

\section{Proof of Lemma~\ref{recurrencelemma2}}\label{appendix2}

For all
$m\in\{1,2,\ldots,h-1\}$ and
$s_m$, $s_{m-1}$, $\ldots$, $s_0\in\{0,1,\ldots,t-1\}$, consider
whether
\begin{eqnarray}
\sum_{r=0}^m\,s_r\cdot t^r\le \sum_{r=0}^m\,s'_r\cdot t^r
\label{notexceedingthenumber}
\end{eqnarray}
as follows:
\begin{enumerate}[(I)]
\item\label{myitem1}
If
$s_m<s'_m$,
then
Eq.~(\ref{notexceedingthenumber}) holds.
\item\label{myitem2}
If
$s_m=s'_m$,
then
Eq.~(\ref{notexceedingthenumber}) holds
if and only if
$(s_{m-1},s_{m-2},\ldots,s_0)$
is the $t$-ary representation of one of $0$, $1$, $\ldots$,
$\sum_{r=0}^{m-1}\,s'_r\cdot t^r$.
\item\label{myitem3}
If
$s_m>s'_m$,
then
Eq.~(\ref{notexceedingthenumber})
fails.
\end{enumerate}
Items~(\ref{myitem1})--(\ref{myitem3})
follow from the basics of
$t$-ary representations.


\begin{lemma}[{cf.~\cite[Lemma~7]{Cha15}}]
\label{summingsimpledistances}
For all $m\in\{1,2,\ldots,h-1\}$,
\begin{eqnarray}
&&\sum_{s_m=0}^{s'_m}\,
\sum_{s_{m-1},s_{m-2},\ldots,s_0=0}^{t-1}\,
\chi\left[\sum_{r=0}^m\, s_r\cdot t^r\le \sum_{r=0}^m\, s'_r\cdot t^r
\right]
\cdot
d^{(n-\sigma)}\left(i,it+s_m\right)\nonumber\\
&=&
\left(1+\sum_{r=0}^{m-1}\,s'_r\cdot t^r\right)\,
d^{(n-\sigma)}\left(i,it+s'_m\right)
+
t^m\sum_{s_m=0}^{s'_m-1}\,d^{(n-\sigma)}
\left(i,it+s_m\right).
\label{sumofsimpledistances}
\end{eqnarray}
\end{lemma}
\begin{proof}
Items~(\ref{myitem2})~and~(\ref{myitem1})
account for the first and the second terms
of the right-hand side of
Eq.~(\ref{sumofsimpledistances}), respectively.\footnote{ 
Note
that there are $t^m$ tuples $(s_{m-1},s_{m-2},\ldots,s_0)\in\{0,1,\ldots,t-1\}^m$,
among which
$1+\sum_{r=0}^{m-1}\,s'_r\cdot t^r$
are $t$-ary
representations of
one of $0$, $1$, $\ldots$, $\sum_{r=0}^{m-1}\,s'_r\cdot t^r$.
} 
\end{proof}

\comment{ 
The following
representation of $f(\cdot,m)$ in terms of $f(\cdot,m-1)$ and $g(\cdot,m-1)$
is laborious but not hard to verify.
}

The
proof of
Lemma~\ref{recurrencelemma2}
just
uses
items~(\ref{myitem1})--(\ref{myitem3})
na{\"\i}vely and with great patience:

\begin{restate2}[{cf.~\cite[Lemma~7]{Cha15}}]
For all $i\in\{0,1,\ldots,n-\sigma-1\}$ and $m\in\{1,2,\ldots,h-1\}$,
\begin{eqnarray}
f\left(i,m\right)
&=&
\left(1+\sum_{r=0}^{m-1}\, s'_r\cdot t^r\right)
d^{(n-\sigma)}
\left(i,it+s'_m\right)\nonumber\\
&+&
t^m
\sum_{s_m=0}^{s'_m-1}\,
d^{(n-\sigma)}
\left(i,it+s_m\right)\nonumber\\
&+&
f\left(it+s'_m
\bmod{\left(n-\sigma\right)},\,
m-1\right)\nonumber\\
&+&\sum_{s_m=0}^{s'_m-1}\, g\left(it+s_m
\bmod{\left(n-\sigma\right)},\,
m-1
\right).
\nonumber
\end{eqnarray}
\end{restate2}
\begin{proof}
\comment{ 
Observe
the following
for all
$s_m$, $s_{m-1}$, $\ldots$, $s_0\in\{0,1,\ldots,t-1\}$:
\begin{enumerate}[(i)]
\item\label{item1}
If $s_m=s'_m$, then
$\sum_{r=0}^m\, s_r\cdot t^r\le \sum_{r=0}^m\, s'_r\cdot t^r$
if and only if
$\sum_{r=0}^{m-1}\, s_r\cdot t^r\le \sum_{r=0}^{m-1}\, s'_r\cdot t^r$;
\item\label{item2}
If $s_m<s'_m$, then
$\sum_{r=0}^m\, s_r\cdot t^r< \sum_{r=0}^m\, s'_r\cdot t^r$;
\item\label{item3}
If $s_m>s'_m$, then
$\sum_{r=0}^m\, s_r\cdot t^r> \sum_{r=0}^m\, s'_r\cdot t^r$.
\end{enumerate}
}
We have
{\footnotesize 
\begin{eqnarray}
&&
f\left(i,m\right)
\label{thethingwewanttorecurseon}\\
&\stackrel{\text{(\ref{subsumlessthanorequalto})}}{=}&
\sum_{s_m=0}^{t-1}\,
\sum_{s_{m-1},s_{m-2},\ldots,s_0=0}^{t-1}\,
\chi\left[\sum_{r=0}^m\, s_r\cdot t^r\le \sum_{r=0}^m\, s'_r\cdot t^r
\right]
\cdot
\left[d^{(n-\sigma)}\left(i,it+s_m\right)
\vphantom{\sum_{\ell=m-2-k}^{m-2}\, s_\ell\cdot t^\ell\bmod{n}}
\right.\nonumber\\
&&\left.+\sum_{k=1}^{m}\, d^{(n-\sigma)}
\left(it^{k}+\sum_{\ell=0}^{k-1}\,
s_{m-\ell}\cdot t^{k-1-\ell},\, it^{k+1}+\sum_{\ell=0}^{k}\,
s_{m-\ell}\cdot t^{k-\ell}
\right)
\right]\nonumber\\
&\stackrel{\text{Lemma~\ref{tedioustwicechangeofvariables}}}{=}&
\sum_{s_m=0}^{t-1}\,
\sum_{s_{m-1},s_{m-2},\ldots,s_0=0}^{t-1}\,
\chi\left[\sum_{r=0}^m\, s_r\cdot t^r\le \sum_{r=0}^m\, s'_r\cdot t^r
\right]
\cdot
\left[d^{(n-\sigma)}\left(i,it+s_m\right)
\vphantom{\sum_{k=0}^{m-1}\, d\left(it^{k+1}+s_m\cdot t^k+\sum_{\ell=0}^{k-1}\,s_{m-1-\ell}\cdot t^{k-1-\ell},\, it^{k+2}+s_m\cdot t^{k+1}+\sum_{\ell=0}^{k}\,s_{m-1-\ell}\cdot t^{k-\ell}\right)}
\right.\nonumber\\
&&\left.+\sum_{k=0}^{m-1}\,
d^{(n-\sigma)}
\left(it^{k+1}+s_m\cdot t^k+\sum_{\ell=0}^{k-1}\,
s_{m-1-\ell}\cdot t^{k-1-\ell},\, it^{k+2}+s_m\cdot t^{k+1}+\sum_{\ell=0}^{k}\,
s_{m-1-\ell}\cdot t^{k-\ell}
\right)
\right]
\nonumber\\
&\stackrel{\text{item~(\ref{myitem3})}}{=}&
\sum_{s_m=0}^{s'_m}\,
\sum_{s_{m-1},s_{m-2},\ldots,s_0=0}^{t-1}\,
\chi\left[\sum_{r=0}^m\, s_r\cdot t^r\le \sum_{r=0}^m\, s'_r\cdot t^r
\right]
\cdot
\left[d^{(n-\sigma)}\left(i,it+s_m\right)
\vphantom{\sum_{k=0}^{m-1}\, d\left(it^{k+1}+s_m\cdot t^k+\sum_{\ell=0}^{k-1}\,s_{m-1-\ell}\cdot t^{k-1-\ell},\, it^{k+2}+s_m\cdot t^{k+1}+\sum_{\ell=0}^{k}\,s_{m-1-\ell}\cdot t^{k-\ell}\right)}
\right.\nonumber\\
&&\left.+\sum_{k=0}^{m-1}\,
d^{(n-\sigma)}
\left(it^{k+1}+s_m\cdot t^k+\sum_{\ell=0}^{k-1}\,
s_{m-1-\ell}\cdot t^{k-1-\ell},\, it^{k+2}+s_m\cdot t^{k+1}+\sum_{\ell=0}^{k}\,
s_{m-1-\ell}\cdot t^{k-\ell}
\right)
\right].
\nonumber
\end{eqnarray}
}
So by Lemma~\ref{summingsimpledistances},
it remains to prove that
{\footnotesize 
\begin{eqnarray}
&&\sum_{s_m=0}^{s'_m}\,
\sum_{s_{m-1},s_{m-2},\ldots,s_0=0}^{t-1}\,
\chi\left[\sum_{r=0}^m\, s_r\cdot t^r\le \sum_{r=0}^m\, s'_r\cdot t^r
\right]
\nonumber\\
&\cdot&
\sum_{k=0}^{m-1}\,
d^{(n-\sigma)}\left(it^{k+1}+s_m\cdot t^k+\sum_{\ell=0}^{k-1}\,
s_{m-1-\ell}\cdot t^{k-1-\ell},\, it^{k+2}+s_m\cdot t^{k+1}+\sum_{\ell=0}^{k}\,
s_{m-1-\ell}\cdot t^{k-\ell}
\right)
\,\,\,\,\,\,\,\,\,\,\,\,
\nonumber\\
&=&
f\left(it+s'_m
\bmod{\left(n-\sigma\right)},\,m-1\right)
+\sum_{s_m=0}^{s'_m-1}\,
g\left(it+s_m
\bmod{\left(n-\sigma\right)},\,m-1
\right).
\label{whatremains}
\end{eqnarray}
}

Separating
the left-hand side of
Eq.~(\ref{whatremains})
according to whether $s_m=s'_m$ or $s_m\le s'_m-1$,
{\footnotesize 
\begin{eqnarray*}
&&\sum_{s_m=0}^{s'_m}\,
\sum_{s_{m-1},s_{m-2},\ldots,s_0=0}^{t-1}\,
\chi\left[\sum_{r=0}^m\, s_r\cdot t^r\le \sum_{r=0}^m\, s'_r\cdot t^r
\right]
\nonumber\\
&\cdot&
\sum_{k=0}^{m-1}\,
d^{(n-\sigma)}
\left(it^{k+1}+s_m\cdot t^k+\sum_{\ell=0}^{k-1}\,
s_{m-1-\ell}\cdot t^{k-1-\ell},\, it^{k+2}+s_m\cdot t^{k+1}+\sum_{\ell=0}^{k}\,
s_{m-1-\ell}\cdot t^{k-\ell}
\right)\nonumber\\
&=&
\sum_{s_{m-1},s_{m-2},\ldots,s_0=0}^{t-1}\,
\chi\left[\sum_{r=0}^{m-1}\, s_r\cdot t^r\le \sum_{r=0}^{m-1}\, s'_r\cdot t^r
\right]
\nonumber\\
&\cdot&
\sum_{k=0}^{m-1}\,
d^{(n-\sigma)}
\left(it^{k+1}+s'_m\cdot t^k+\sum_{\ell=0}^{k-1}\,
s_{m-1-\ell}\cdot t^{k-1-\ell},\, it^{k+2}+s'_m\cdot t^{k+1}+\sum_{\ell=0}^{k}\,
s_{m-1-\ell}\cdot t^{k-\ell}
\right)\nonumber\\
&+&
\sum_{s_m=0}^{s'_m-1}\,
\sum_{s_{m-1},s_{m-2},\ldots,s_0=0}^{t-1}\,
\chi\left[\sum_{r=0}^m\, s_r\cdot t^r\le \sum_{r=0}^m\, s'_r\cdot t^r\right]
\nonumber\\
&\cdot&
\sum_{k=0}^{m-1}\,
d^{(n-\sigma)}
\left(it^{k+1}+s_m\cdot t^k+\sum_{\ell=0}^{k-1}\,
s_{m-1-\ell}\cdot t^{k-1-\ell},\, it^{k+2}+s_m\cdot t^{k+1}+\sum_{\ell=0}^{k}\,
s_{m-1-\ell}\cdot t^{k-\ell}
\right)\nonumber\\
&\stackrel{\text{item~(\ref{myitem1})}}{=}&
\sum_{s_{m-1},s_{m-2},\ldots,s_0=0}^{t-1}\,
\chi\left[\sum_{r=0}^{m-1}\, s_r\cdot t^r\le \sum_{r=0}^{m-1}\, s'_r\cdot t^r
\right]
\nonumber\\
&\cdot&
\sum_{k=0}^{m-1}\,
d^{(n-\sigma)}
\left(it^{k+1}+s'_m\cdot t^k+\sum_{\ell=0}^{k-1}\,
s_{m-1-\ell}\cdot t^{k-1-\ell},\, it^{k+2}+s'_m\cdot t^{k+1}+\sum_{\ell=0}^{k}\,
s_{m-1-\ell}\cdot t^{k-\ell}
\right)\nonumber\\
&+&
\sum_{s_m=0}^{s'_m-1}\,
\sum_{s_{m-1},s_{m-2},\ldots,s_0=0}^{t-1}\,
\nonumber\\
&&
\sum_{k=0}^{m-1}\,
d^{(n-\sigma)}
\left(it^{k+1}+s_m\cdot t^k+\sum_{\ell=0}^{k-1}\,
s_{m-1-\ell}\cdot t^{k-1-\ell},\, it^{k+2}+s_m\cdot t^{k+1}+\sum_{\ell=0}^{k}\,
s_{m-1-\ell}\cdot t^{k-\ell}
\right)
\nonumber\\
&\stackrel{\text{(\ref{themodularversionofdistances})}}{=}&
\sum_{s_{m-1},s_{m-2},\ldots,s_0=0}^{t-1}\,
\chi\left[\sum_{r=0}^{m-1}\, s_r\cdot t^r\le \sum_{r=0}^{m-1}\, s'_r\cdot t^r
\right]
\nonumber\\
&\cdot&
\sum_{k=0}^{m-1}\,
d^{(n-\sigma)}
\left(\left(it+s'_m
\bmod{\left(n-\sigma\right)}
\right) t^k+\sum_{\ell=0}^{k-1}\,
s_{m-1-\ell}\cdot t^{k-1-\ell},
\right.\nonumber\\
&&\left.
\, \left(it+s'_m
\bmod{\left(n-\sigma\right)}
\right) t^{k+1}+\sum_{\ell=0}^{k}\,
s_{m-1-\ell}\cdot t^{k-\ell}
\right)\nonumber\\
&+&
\sum_{s_m=0}^{s'_m-1}\,
\sum_{s_{m-1},s_{m-2},\ldots,s_0=0}^{t-1}\,
\nonumber\\
&&
\sum_{k=0}^{m-1}\,
d^{(n-\sigma)}
\left(\left(it+s_m
\bmod{\left(n-\sigma\right)}
\right) t^k+\sum_{\ell=0}^{k-1}\,
s_{m-1-\ell}\cdot t^{k-1-\ell},
\right.\nonumber\\
&&\left.
\, \left(it+s_m
\bmod{\left(n-\sigma\right)}
\right) t^{k+1}+\sum_{\ell=0}^{k}\,
s_{m-1-\ell}\cdot t^{k-\ell}
\right)
\nonumber\\
&\stackrel{\text{(\ref{subsumlessthanorequalto})--(\ref{subsumlessthan})}}{=}&
f\left(it+s'_m
\bmod{\left(n-\sigma\right)},\,m-1\right)
+\sum_{s_m=0}^{s'_m-1}\,
g\left(it+s_m
\bmod{\left(n-\sigma\right)},\,m-1
\right).
\end{eqnarray*}
}
\end{proof}

\bibliographystyle{plain}
\bibliography{median_shaving_small_factor}

\noindent

\end{document}